\documentclass[11pt]{article}
\usepackage[T1]{fontenc}
\usepackage[utf8]{inputenc}
\usepackage[a4paper,margin=25mm]{geometry}
\usepackage{amsmath,amssymb,amsthm}
\usepackage{graphicx}
\usepackage{booktabs}
\usepackage{xcolor}
\usepackage[numbers,sort&compress]{natbib}
\usepackage[colorlinks=true,linkcolor=blue!50!black,citecolor=blue!50!black,urlcolor=blue!50!black]{hyperref}
\usepackage{cleveref}
\graphicspath{{figures/}}
\newtheorem{proposition}{Proposition}
\newtheorem{corollary}{Corollary}
\theoremstyle{definition}

\theoremstyle{remark}
\newtheorem{remark}{Remark}
\newcommand{\KC}{K_{\mathrm{C}}}
\newcommand{\KD}{K_{\mathrm{D}}}

\usepackage{authblk}

\title{Beyond the Clique: Comparing Clique and Dowker Complexes\\ for Co-occurrence Data in Learning Analytics}
\author[1]{Koichi Yasutake\thanks{Corresponding author, Email: \texttt{ystake@hiroshima-u.ac.jp}}}
\author[2]{Wakana Tsuji\thanks{Email: \texttt{wakana.gen@mail.utoronto.ca}}}
\author[3]{Hitoshi Inoue\thanks{Email: \texttt{jinoue@nakamura-u.ac.jp}}}
\affil[1]{Graduate School of Humanities and Social Sciences, Hiroshima University, Japan}
\affil[2]{Ontario Institute for Studies in Education, University of Toronto, Canada}
\affil[3]{Faculty of Business, Marketing and Distribution, Nakamura Gakuen University, Japan}
\date{}

\begin{document}
\maketitle

\begin{abstract}
Learning analytics increasingly represents the relations in co-occurrence data---codes in a window of discourse, participants in a thread, tags on a post---as simplicial complexes and analyses them with persistent homology.
The standard approach uses the \emph{clique complex}. Because its simplices are determined by the pairwise network alone, it cannot distinguish ``three elements co-occurred together'' from ``each of the three pairs co-occurred separately''.
The \emph{Dowker complex}, in contrast, takes as simplices the sets of entities actually observed to co-occur. We compare the two.
First, we show that, on the same 1-skeleton, the Dowker complex is a subcomplex of the clique complex, the map on first homology induced by the inclusion is surjective, and its kernel is generated by \emph{phantom} triangles that never co-occurred; that is, the clique construction can only erase holes. We also give a criterion for agreement that can be checked on triples of groups.
We then measure the loss on real data and find that it takes two forms: a \emph{permanent} loss, in which phantom triangles erase holes for good, and a \emph{transient} loss, in which holes are filled before the corresponding triples have co-occurred. On four Stack Exchange data sets, 46--91\% of clique triangles are phantom and 522--1{,}970 holes are erased, and against a degree-preserving null model the two constructions lead to different conclusions.
On the example data of the learning-analytics R packages \texttt{tna}/\texttt{Nestimate}, no triangle is phantom, yet the clique complex fills every hole before the Dowker complex does, so its first Betti number $\beta_1$ is 0 at every threshold while the Dowker complex detects holes. Which form dominates depends on the data, and neither is visible at the fixed thresholds used in practice; we therefore recommend reporting the number of classes the clique complex fills earlier together with the phantom rate and the kernel dimension.
We conclude that the construction should follow the data type and that, for observed groups, the Dowker complex is the appropriate choice.
Code is available at \url{https://github.com/igu-lab/beyond-the-clique}.

\end{abstract}

\section{Introduction}\label{sec:intro}

Network methods have become a standard lens in learning analytics.
Transition Network Analysis (TNA) models sequences of learning actions as weighted directed networks~\citep{saqr2025tna,tikka2025tna}, and its companion package \texttt{Nestimate} adds higher-order tools, including higher-order Markov models, path-anomaly detection, hypergraphs, simplicial complexes and persistent homology~\citep{saqr2026nestimate,nestimate085}.
Heterogeneous Interaction Network Analysis (HINA) models learners and their coded behaviours as bipartite networks~\citep{feng2026hina}.

We have previously argued that network analyses in the Learning Sciences, grounded in classical graph theory, can represent only pairwise relations, and proposed temporal hypergraphs as a framework for analysing many-body interactions in learning as they unfold over time~\citep{kaz2025}.
Extending the period segmentation of \citet{cristea2025} (halves and quarters), we also proposed deriving data-driven period boundaries from change points in the connected components ($\beta_{0}$) of zigzag persistent homology~\citep{inoue2026}.
That work, however, tracked only the temporal evolution of connected components and did not address the structure of many-body co-occurrence.

Note that ``higher-order'' structure in data has two distinct meanings~\citep{lambiotte2019networks,battiston2020networks}.
\emph{Higher-order dependencies} (HOD) concern memory in sequences: whether the next state depends on more than the current one~\citep{xu2016representing}.
\emph{Higher-order interactions} (HOI) concern many-body co-occurrence: whether three or more entities act together in a way that pairs cannot represent.
\texttt{Nestimate} covers HOD extensively. This paper is about HOI.

The question of this paper is simple: can an analysis tell ``A, B and C talked together once'' from ``A and B, B and C, and A and C each talked once'' (\Cref{fig:core})?
Both produce the same pairwise network.
The clique complex, which fills a triangle whenever its three edges are present, also cannot tell them apart; the distinction is well known in network science~\citep{giusti2016two,patania2017shape,bick2023what}, but it has not reached the simplicial pipelines used in learning analytics.
The Dowker complex of the observed groups~\citep{dowker1952,chowdhury2018dowker} can: it fills the triangle in the first case and leaves a hole in the second.
A hypergraph can also distinguish the two, but the Dowker complex is a simplicial complex built directly from the observed groups (hyperedges), so homology and persistent homology apply to it as they stand, and it can be compared with the clique complex on equal terms.
In this sense, the Dowker complex carries the hypergraph perspective we proposed earlier~\citep{kaz2025} over to topological analysis.

\begin{figure}[t]
  \centering
  \includegraphics[width=\linewidth]{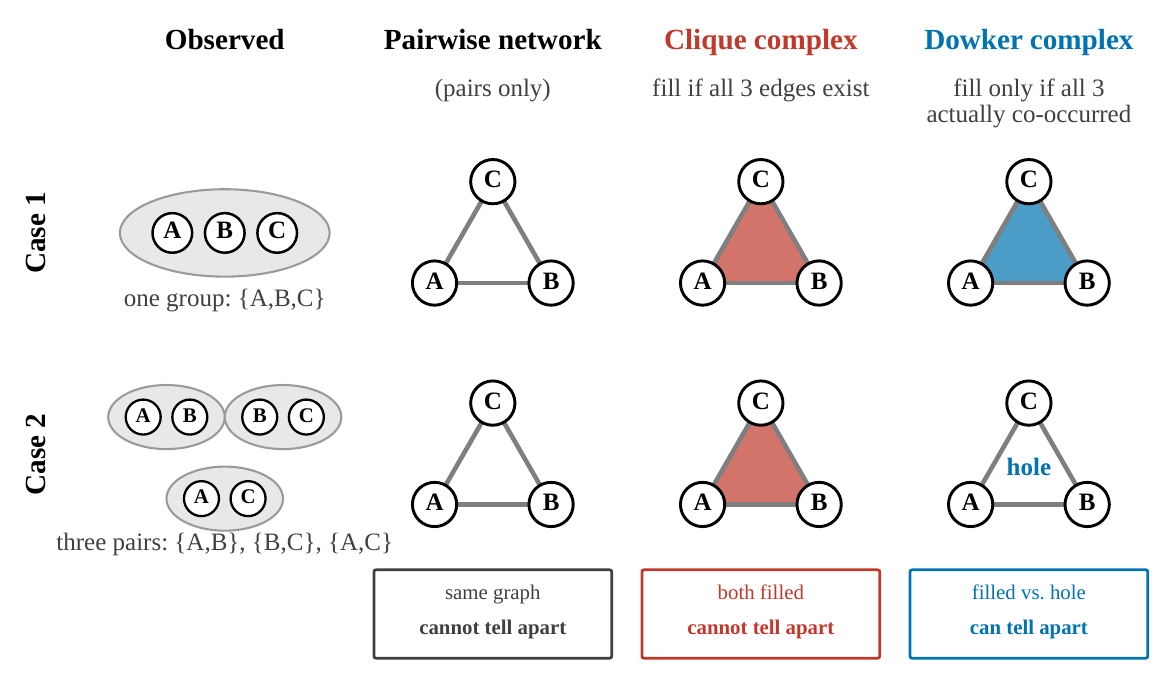}
  \caption{The core distinction. The pairwise network and the clique complex are identical for a single three-person group (Case~1) and for three separate pairs (Case~2). The Dowker complex fills the triangle only when the three actually co-occurred, leaving a hole in Case~2.}
  \label{fig:core}
\end{figure}

Current simplicial pipelines in learning analytics build clique complexes from pairwise matrices whatever the data type. For example, \texttt{build\_simplicial} and \texttt{persistent\_homology} in \texttt{Nestimate}, the companion package of TNA, fill every clique of a thresholded weight matrix.
The package can already represent observed groups as a hypergraph (\texttt{bipartite\_groups}).
Taking the position that the representation should follow the system being studied~\citep{torres2021why}, we examine what the choice of construction entails.

We make three contributions.
\begin{enumerate}
  \item \textbf{A construction rule.} We state which complex matches which data type (point clouds, observed groups, transitions, time-stamped groups) and show that, for observed groups, the clique complex is the wrong construction (\Cref{sec:rules}).
  \item \textbf{A measure of the loss.} We make explicit that the clique complex can only erase first-homology classes of the Dowker complex, identify the erased classes with \emph{phantom} triangles, and distinguish two forms of loss: a permanent loss, measured by the phantom rate and the kernel dimension, and a transient loss, measured by the number of classes the clique complex fills earlier. Unlike the metric setting, where the clique (Dowker--Rips) complex is a provably good approximation~\citep{huber2025flagifying}, count filtrations of co-occurrence data admit no such guarantee. We further give a local criterion for when the two constructions agree (\Cref{sec:theory}).
  \item \textbf{Quantification on real data and real tools.} On four Stack Exchange data sets and on the two data sets distributed with \texttt{tna}/\texttt{Nestimate}, reproducing the tools' output exactly, we measure what is lost and compare both constructions with a degree-preserving null model (\Cref{sec:results}).
\end{enumerate}
\Cref{app:tools} relates these results to the current state of the TNA toolchain, including a window co-occurrence complex added to \texttt{hypernets} on 1 October 2026, and to transition flows.

\section{Background}\label{sec:background}

This section fixes the notation used throughout: observed groups and their co-occurrence counts, the two complexes we compare (clique and Dowker), and their filtrations.

\paragraph{Observed groups.}
Let $V$ be a finite set of entities (codes, participants, tags) and let $\mathcal{H}=(h_1,\dots,h_m)$ be the observed groups, each $h_i\subseteq V$ (a window of coded discourse, a thread, a post).
Equivalently, $\mathcal{H}$ is the bipartite relation $R\subseteq V\times\{1,\dots,m\}$ with $(v,i)\in R$ whenever $v\in h_i$ (used in \Cref{sec:theory}).
For $\sigma\subseteq V$ let $w(\sigma)=|\{i:\sigma\subseteq h_i\}|$ be its co-occurrence count.

\paragraph{Clique complex.}
Given a graph $G$, its clique (flag) complex contains every vertex set that is pairwise adjacent in $G$.
It is determined by $G$ alone.
The Vietoris--Rips complex of a point cloud is the clique complex of its distance-threshold graph and is the standard construction for metric data~\citep{carlsson2009topology}.

\paragraph{Dowker complex.}
The Dowker complex of $\mathcal{H}$ contains $\sigma$ if and only if $\sigma\subseteq h_i$ for some $i$, i.e.\ $w(\sigma)\ge 1$~\citep{dowker1952}.
Its simplices are exactly the sets of entities that were observed together (and their subsets).

\paragraph{Count filtration.}
For $t\ge 1$ define
\begin{align}
  \KD(t) &= \{\sigma\subseteq V : w(\sigma)\ge t\}, \label{eq:KD}\\
  G(t)   &= \text{the graph on } V \text{ with edges } \{u,v\} \text{ such that } w(\{u,v\})\ge t, \\
  \KC(t) &= \text{the clique complex of } G(t). \label{eq:KC}
\end{align}
Since $w$ is non-increasing under inclusion, $\KD(t)$ is a simplicial complex, and decreasing $t$ gives nested filtrations of both complexes; strong co-occurrences enter first.
$\KC(t)$ is the construction used by \texttt{Nestimate::build\_simplicial(type = "clique", threshold = t)} applied to a co-occurrence matrix (\Cref{sec:rules}).
We compute persistent homology~\citep{edelsbrunner2002topological,zomorodian2005computing} over $\mathbb{F}_2$.

\section{Related work}\label{sec:related}

We review three lines of work --- complexes built from observed groups in network science, clique versus Dowker constructions in topology, and topological and higher-order methods in learning analytics --- and position this paper among them.

\paragraph{Complexes from observed groups in network science.}
Building simplicial complexes directly from observed groups, rather than from cliques of the pairwise network, is standard in parts of network science: papers as simplices of their authors~\citep{patania2017shape}, concepts co-occurring in mathematical texts~\citep{salnikov2018cooccurrence}, and simplicial closure in time-stamped group data~\citep{benson2018simplicial}.
Reviews discuss when a clique complex adds interactions that were never observed~\citep{torres2021why,bick2023what}.

\paragraph{Clique and Dowker constructions in topology.}
On the topological side, \citet{huber2025flagifying} study the flagification of Dowker complexes (the Dowker--Rips complex) as a cheaper approximation in the metric setting.
We contrast this with count filtrations of co-occurrence data in \Cref{sec:theory}.

\paragraph{Topological and higher-order methods in learning analytics.}
In learning analytics, topological data analysis has mainly been introduced through Vietoris--Rips complexes of point clouds and Mapper~\citep{munch2017user}; the simplicial tools of \texttt{Nestimate}~\citep{saqr2026nestimate} use clique complexes of pairwise matrices.
A window co-occurrence complex~\citep{salnikov2018cooccurrence} has just been added to the development version of \texttt{hypernets}~\citep{hypernets060}; it coincides with the Dowker complex only at the lowest threshold and is not compared with the clique construction (\Cref{app:hypernets}).
To our knowledge, the loss incurred by the clique construction on co-occurrence data has not previously been characterized or quantified.

\section{What the clique complex erases}\label{sec:theory}

Throughout, the vertex set of $G(t)$ is $V_t=\{v: w(\{v\})\ge t\}$, so that $G(t)$ is exactly the 1-skeleton of $\KD(t)$ and $\KC(t)$ is the flag complex of that 1-skeleton.
In the terminology of \citet{huber2025flagifying}, $\KC(t)$ is the \emph{Dowker--Rips complex}, the flagification of the Dowker complex.
The following facts are elementary; we state them because they turn the difference between the two constructions into quantities that can be measured on data.

\begin{proposition}[The clique complex can only erase first homology]\label{prop:erase}
For every $t\ge1$:
\begin{enumerate}
  \item $\KD(t)$ and $\KC(t)$ have the same 1-skeleton $G(t)$, and $\KD(t)\subseteq\KC(t)$.
  \item The inclusion induces a surjection $\iota_*:H_1(\KD(t))\to H_1(\KC(t))$ (over any field); in particular $\beta_1^{\mathrm{C}}(t)\le\beta_1^{\mathrm{D}}(t)$.
  \item $\ker\iota_*$ is spanned by the classes $[\partial\tau]$ of the \emph{phantom triangles} $\tau\in\KC(t)\setminus\KD(t)$, i.e.\ triangles whose three edges co-occur at least $t$ times but whose three vertices co-occur fewer than $t$ times. Hence $\beta_1^{\mathrm{D}}(t)-\beta_1^{\mathrm{C}}(t)=\dim\ker\iota_*\le$ the number of phantom triangles.
\end{enumerate}
\end{proposition}

\begin{proof}
(1) Both complexes have vertex set $V_t$ and edge set $\{\{u,v\}:w(\{u,v\})\ge t\}$. If $\sigma\in\KD(t)$, every pair in $\sigma$ has $w\ge w(\sigma)\ge t$, so $\sigma$ is a clique of $G(t)$.
(2) For any complex $K$, $H_1(K)=Z_1(K)/B_1(K)$ where the cycle space $Z_1$ depends only on the 1-skeleton. The two complexes share their 1-skeleton, so $Z_1(\KD)=Z_1(\KC)=:Z_1$, and $B_1(\KD)\subseteq B_1(\KC)$ because $\KD\subseteq\KC$. Thus $\iota_*$ is the quotient map $Z_1/B_1(\KD)\to Z_1/B_1(\KC)$, which is surjective.
(3) $\ker\iota_*=B_1(\KC)/B_1(\KD)$. $B_1(\KC)$ is spanned by boundaries of triangles of $\KC$; those of triangles in $\KD$ vanish in the quotient, leaving the boundaries of phantom triangles.
\end{proof}

\begin{remark}[Persistence]
Because $\KD(t)\subseteq\KC(t)$ for every $t$ with a common 1-skeleton, every first-homology class is born at the same filtration value in both complexes and dies no later in $\KC$ than in $\KD$.
Classes therefore fall into four types: filled at the same step, filled earlier by the clique complex, filled only by the clique complex (the kernel), and filled in neither (\Cref{fig:classes}).
\end{remark}

\begin{remark}[Approximation guarantees do not transfer to count filtrations]\label{rem:interleave}
When the relation comes from distances, \citet{huber2025flagifying} show that the Dowker--Rips filtration is multiplicatively 3-interleaved with the Dowker filtration ($D_\varepsilon\subseteq DR_\varepsilon\subseteq D_{3\varepsilon}$), which justifies using it as a cheaper drop-in replacement. Co-occurrence counts carry no such metric structure. A cycle of three pairs that never co-occur as a triple is filled in $\KC(t)$ as soon as its edges appear, but, unless other observed groups fill it, it stays open in $\KD(t)$ down to $t=1$; no rescaling of $t$ closes it. The kernel classes reported in \Cref{sec:results} are of exactly this kind.
\end{remark}

By \Cref{prop:erase}, the entire loss caused by the clique construction appears in the phantom triangles and in $\ker\iota_*$, and by \Cref{rem:interleave}, nothing in a count filtration guarantees that this loss is small, as it is in the metric case. The loss therefore has to be measured on each data set. We define the measures used in \Cref{sec:results} as follows.

\paragraph{Measures of loss.}
The \emph{phantom rate} at $t$ is $|\KC(t)_2\setminus\KD(t)_2|/|\KC(t)_2|$, the share of clique triangles that never co-occurred $t$ times. The \emph{kernel dimension} is $\dim\ker\iota_*$. Along the filtration we report the number of classes filled only by the clique complex and the number filled earlier.
These correspond to two forms of loss. The \emph{permanent} loss consists of the kernel classes, erased by triangles that are phantom even at $t=1$, i.e.\ whose triples never co-occurred. The \emph{transient} loss consists of the classes filled earlier: the triangles that fill them are phantom only at intermediate $t$, where the pairs have co-occurred $t$ times but the triple has not. A complex whose phantom rate at $t=1$ is zero can still have a large transient loss.

\medskip
\Cref{prop:erase} concerns the loss between the two constructions on a fixed vertex set. The clique construction has a second problem: its result depends on which vertex set is chosen in the first place.
Because observed groups are also a bipartite relation $R$ (\Cref{sec:background}), complexes can be built from the same data on the group side as well as on the entity side. In learning data this is the choice between taking learners as vertices and taking the conversations or threads they joined as vertices, a choice usually left to the analyst.
We therefore ask whether the result depends on that choice, for each of the two constructions.

\paragraph{Two views of a bipartite relation.}
The relation $R\subseteq V\times\{1,\dots,m\}$ of \Cref{sec:background} is the case $X=V$, $Y=\{1,\dots,m\}$ below.
\begin{proposition}[Duality: all dimensions for Dowker, only up to $H_1$ for clique complexes]\label{prop:duality}
Let $R\subseteq X\times Y$ be a filtered relation.
(i) The Dowker complexes on $X$ and on $Y$ have isomorphic persistent homology in every dimension~\citep{dowker1952,chowdhury2018dowker}.
(ii) Their flagifications, i.e.\ the clique complexes of the two one-mode projections, have isomorphic persistent homology in dimensions~0 and~1, but not in general in dimensions $\ge2$ (\citealp[Thm.~1.8 with $k=2$, Prop.~3.4]{huber2025flagifying}).
\end{proposition}
Thus the choice of side does not matter for Dowker complexes, whereas for clique complexes it can change the higher-dimensional topology.
The counterexample of \citet[Prop.~3.4]{huber2025flagifying}, the vertices and edge midpoints of a regular tetrahedron, reads naturally as learning data: four learners and six two-person conversations, one for every pair.
On the learner side the Dowker complex is the complete graph $K_4$ ($\beta_1=3$) and the clique complex is a solid tetrahedron: the clique construction asserts that all four worked together.
On the conversation side the Dowker complex again has $\beta_1=3$, while the clique complex is the boundary of an octahedron ($\beta_1=0$, $\beta_2=1$), a two-dimensional void that exists in neither the data nor the other view.
Both clique complexes have $\beta_1=0$, as (ii) requires, and both have erased the three cycles of the data (\Cref{prop:erase}).
Details are given in \Cref{app:duality}.

\paragraph{When do the two constructions agree?}
\Cref{prop:erase} measures the loss once both complexes have been computed. The next statement says, before any homology is computed, exactly when there is no loss, and locates the configurations of groups that produce it.
A hypergraph is \emph{conformal} if every clique of its 2-section (the graph joining two vertices whenever some hyperedge contains both) is contained in a hyperedge~\citep{berge1989hypergraphs}.

\begin{proposition}[A local criterion for agreement]\label{prop:conformal}
For $t\ge1$ let $\mathcal{F}_t$ be the set of maximal simplices (facets) of $\KD(t)$, viewed as a hypergraph on $V_t$.
\begin{enumerate}
  \item $\KC(t)=\KD(t)$ if and only if $\mathcal{F}_t$ is conformal.
  \item Equivalently (Gilmore's criterion~\citep{berge1989hypergraphs}), $\KC(t)=\KD(t)$ if and only if for any three facets $F_1,F_2,F_3\in\mathcal{F}_t$ the set $(F_1\cap F_2)\cup(F_2\cap F_3)\cup(F_1\cap F_3)$ is contained in some facet.
  \item A triangle $\{a,b,c\}$ of $G(t)$ is phantom if and only if no facet contains it. In particular, at $t=1$ a phantom triangle requires three distinct observed groups that contain $\{a,b\}$, $\{b,c\}$ and $\{a,c\}$, respectively, while no group contains all three.
\end{enumerate}
\end{proposition}

\begin{proof}
The 2-section of $\mathcal{F}_t$ is the 1-skeleton of $\KD(t)$, which is $G(t)$ (\Cref{prop:erase}). A vertex set lies in some facet if and only if it is a simplex of $\KD(t)$. Hence $\mathcal{F}_t$ is conformal if and only if every clique of $G(t)$ is a simplex of $\KD(t)$, i.e.\ $\KC(t)\subseteq\KD(t)$; with $\KD(t)\subseteq\KC(t)$ this gives (1). (2) is Gilmore's characterization of conformal hypergraphs applied to $\mathcal{F}_t$. (3) restates the definition of a phantom triangle; at $t=1$ the facets are the maximal observed groups, and if a single group contained two of the three pairs it would contain the triangle.
\end{proof}

\Cref{prop:conformal} turns the question ``does the choice of construction matter for these data?'' into a check on triples of groups that needs no homology computation.
It also says where the loss comes from: the clique construction fails exactly where groups overlap pairwise without any group containing their overlaps---in learning data, three conversations that pairwise share participants or codes, none of which brought the three together.
Laminar collections of groups (any two are disjoint or nested) satisfy the criterion at every threshold, whereas, when any two groups share at most one entity, every three groups that pairwise share three different entities produce a phantom triangle. Real group data lie between these extremes, which is why the size of the loss has to be measured (\Cref{sec:results}). Turning the criterion into quantitative bounds on the phantom rate in terms of the group-size distribution and overlap structure is part of the decision rules of \Cref{sec:outlook}.
Note that $\KC(t)\neq\KD(t)$ does not imply that phantom triangles exist: the clique construction can also fill a hollow tetrahedron all of whose triangles were observed. Only phantom triangles, however, affect $H_1$ (\Cref{prop:erase}).

\section{Which construction for which data}\label{sec:rules}

\begin{table}[t]
\centering
\caption{Matching the complex to the data type.}
\label{tab:rules}
\small
\begin{tabular}{@{}p{0.27\linewidth}p{0.3\linewidth}p{0.36\linewidth}@{}}
\toprule
Data type & Appropriate construction & Clique complex of pairs is\ldots \\
\midrule
Point cloud / distances & Vietoris--Rips (a clique complex) & appropriate \\
Observed groups / bipartite relation & Dowker complex & inappropriate: fills phantom triangles \\
Transitions (directed) & order-preserving (directed flag, path) & inappropriate: symmetrization loses direction \\
Time-stamped groups & cumulative or zigzag Dowker & inappropriate for both reasons \\
\bottomrule
\end{tabular}
\end{table}

A clique complex is the right object when the data \emph{are} pairwise: for a point cloud, pairwise distances are all there is, and the Vietoris--Rips complex is the standard choice~\citep{carlsson2009topology}.
When the data are observed groups, the groups themselves are the primary information and the pairwise network is a lossy projection of it~\citep{torres2021why}; building a clique complex from that projection reconstructs higher-order structure from the very information that cannot contain it (\Cref{prop:erase}).
\Cref{tab:rules} summarizes the correspondence.

\medskip
How do these rules apply to the tools actually used in learning analytics? Taking \texttt{Nestimate}, the companion package of TNA, as an example, we first check what its simplicial functions build, showing that the clique construction is applied to observed-group data as well. We then show that the directional information in TNA data is lost along the way (the ``Transitions (directed)'' row of \Cref{tab:rules}).

\paragraph{What \texttt{Nestimate} builds.}
\texttt{build\_simplicial(type = "clique")} symmetrizes a weight matrix $W$ by the elementwise maximum $\max(W,W^\top)$, keeps the edges whose weight is positive and at least the threshold, and enumerates cliques up to size \texttt{max\_dim}$+1$.
Two other types exist: \texttt{"pathway"} turns the higher-order paths found by \texttt{build\_hon} or \texttt{build\_hypa} into simplices, and \texttt{"vr"} builds a Vietoris--Rips complex from a distance matrix.
\texttt{persistent\_homology} applies the clique construction along a sequence of thresholds.
Observed groups can be represented in the same package as a hypergraph (\texttt{bipartite\_groups}), but in version 0.8.5 this object cannot be passed to \texttt{build\_simplicial}, which stops with ``Cannot extract adjacency matrix from `net\_hypergraph'\,'' (checked on 29 September 2026).
We re-implemented the clique construction and reproduced the package output exactly: for both TNA data sets and all four thresholds used below, the simplices are identical.

\paragraph{Direction.}
TNA models are directed: the probability of a transition from state $a$ to state $b$ generally differs from that of a transition from $b$ to $a$, and in the two TNA data sets this difference is large.
Among the pairs with transitions in both directions, 69\% (\texttt{group\_regulation}) and 42\% (\texttt{human\_long}) have one direction at least twice as likely as the other.
The antisymmetric share $\|W-W^\top\|_F/\|W+W^\top\|_F$ (off-diagonal), which measures the asymmetry of the whole matrix (0 if symmetric, 1 if every pair has transitions in one direction only), is 0.66 and 0.55, so direction accounts for a major part of the transition-probability matrices.
Symmetrization by the elementwise maximum discards this information before any complex is built.
Constructions that keep direction, such as directed flag complexes~\citep{lutgehetmann2020computing} and path homology~\citep{grigoryan2012homologies}, are left to future work; in this paper we treat co-occurrence only.

\section{Data and methods}\label{sec:data}

\Cref{sec:theory} showed that the clique complex can only erase first homology and that the erased classes come from phantom triangles (\Cref{prop:erase}). The proposition locates the loss but does not fix its size, which depends on the data.
In this and the next section we measure the loss on real data: Stack Exchange co-occurrence data, where the group structure is explicit, and the data distributed with \texttt{tna}, which are used in learning analytics practice. For each data set we compare the two constructions and set the difference against the level expected by chance under a null model.

\paragraph{Data.}
\emph{Stack Exchange co-occurrence data}: we built four data sets from the Stack Exchange data dump of 31 December 2025~\citep{stackexchange2025dump}, following the definitions of \citet{benson2018simplicial}.
In \texttt{tags-matheducators} (3{,}157 groups, 2014--2025), \texttt{tags-academia} (39{,}019) and \texttt{tags-physics} (223{,}639, 2010--2025) a group is the set of tags on a question (at most five); in \texttt{threads-matheducators} (3{,}572 groups) a group is the asker together with the answerers of a question.
Math Educators is a community of teachers discussing mathematics teaching. Stack Exchange content is licensed under CC BY-SA.

\emph{TNA data}: \texttt{tna::group\_regulation}, synthetic data distributed with \texttt{tna}: 2{,}000 sequences and 27{,}533 events coded into nine categories of regulation in collaborative learning (e.g.\ \emph{plan}, \emph{monitor}, \emph{consensus}; median sequence length 14); and \texttt{Nestimate::human\_long}, real data from human--AI pair-programming sessions~\citep{saqr2026vibe}: 10{,}796 human turns in 429 sessions, in which nine action codes occur (e.g.\ \emph{command}, \emph{inquire}, \emph{verify}; median length 18).
A group is the set of distinct codes in a sliding window of $w$ consecutive events within one sequence; $w=3$ yields 23{,}625 and 9{,}948 windows.
Windows can contain a repeated code, so some groups have fewer than $w$ elements.
Results for $w=4$ are in the Appendix.
We use windows rather than whole sequences because \texttt{build\_network(method = "co\_occurrence")} counts pairs over the whole sequence regardless of its \texttt{window\_size} argument (\Cref{sec:results}).

\paragraph{Comparison.}
For each data set we build $\KD(t)$ and $\KC(t)$ on the same 1-skeleton (\Cref{eq:KD,eq:KC}) and compute $H_0$ and $H_1$ persistence over $\mathbb{F}_2$ by column reduction, with no cap on group size.
For every $H_1$ class born in the common 1-skeleton we record its death in each complex and classify it as in \Cref{fig:classes}.

\paragraph{Null model.}
We compare each statistic with a null model that preserves the size of every group and the number of groups each entity belongs to, sampled by repeated double swaps of entities between groups (rejecting swaps that create duplicates), $10\times$ the number of incidences per sample.
Each of the $N=200$ samples is drawn independently from the observed data.
We report $z=(\text{observed}-\text{null mean})/\text{null s.d.}$ and empirical one-sided $p$-values, whose minimum attainable value is $1/(N+1)\approx0.005$.
To compare curves across data sets we index the filtration by the fraction of 1-skeleton edges present.

\paragraph{Reproducibility.}
All computations use Python~3 with NumPy; the \texttt{tna} (1.2.3) and \texttt{Nestimate} (0.8.5) outputs were produced in R~4.6.1.
Code and data availability is described at the end of the paper.

\section{Results}\label{sec:results}

\begin{table}[t]
\centering
\caption{Clique versus Dowker complexes on the same 1-skeleton. Phantom rate: share of clique triangles whose three vertices never co-occurred. Max $\beta_1$: maximum of $\beta_1$ along the filtration. $H_1$ classes: among the first-homology classes born on the common 1-skeleton, the number filled only by the clique complex (kept open to the end by the Dowker complex) and the number filled by both but earlier by the clique complex.}
\label{tab:main}
\small
\begin{tabular}{@{}lrrrrrrr@{}}
\toprule
 & & & & \multicolumn{2}{c}{max $\beta_1$} & \multicolumn{2}{c}{$H_1$ classes} \\
\cmidrule(lr){5-6}\cmidrule(l){7-8}
Data & $|V|$ & $|E|$ & \shortstack{phantom\\rate} & clique & Dowker & \shortstack{filled only by\\clique (kernel)} & \shortstack{filled earlier\\by clique} \\
\midrule
tags-matheducators    & 227     & 3{,}309  & 0.81 & 3  & 680      & 522     & 1{,}015  \\
threads-matheducators & 3{,}796 & 34{,}583 & 0.46 & 49 & 2{,}012  & 1{,}908 & 2{,}710  \\
tags-academia         & 470     & 18{,}676 & 0.88 & 5  & 4{,}069  & 1{,}443 & 8{,}661  \\
tags-physics          & 901     & 76{,}455 & 0.91 & 19 & 10{,}951 & 1{,}970 & 39{,}456 \\
\midrule
group\_regulation ($w{=}3$) & 9 & 36 & 0.00 & 0 & 11 & 0 & 28 \\
human\_long ($w{=}3$)       & 9 & 36 & 0.00 & 0 & 23 & 0 & 28 \\
\bottomrule
\end{tabular}
\end{table}

\begin{table}[t]
\centering
\caption{Comparison with a degree- and size-preserving null model ($N=200$): phantom rate of the clique complex (observed and null mean) and $z$-scores of the maximum $\beta_1$. Negative values mean that the observed groups produce fewer cycles than randomized groups. n/a: $\beta_1=0$ in the data and in every null sample, so $z$ is undefined.}
\label{tab:null}
\small
\begin{tabular}{@{}lrrrr@{}}
\toprule
 & \multicolumn{2}{c}{phantom rate} & \multicolumn{2}{c}{$z$ of max $\beta_1$} \\
\cmidrule(lr){2-3}\cmidrule(l){4-5}
Data & observed & null mean & clique complex & Dowker complex \\
\midrule
tags-matheducators    & 0.809 & 0.869 & 1.4    & $-13.4$  \\
threads-matheducators & 0.464 & 0.534 & $-0.5$ & $-11.6$  \\
tags-academia         & 0.883 & 0.919 & 15.0   & $-53.4$  \\
tags-physics          & 0.911 & 0.955 & 135.2  & $-270.1$ \\
\midrule
group\_regulation ($w{=}3$) & 0.000 & 0.001 & n/a & 0.9     \\
human\_long ($w{=}3$)       & 0.000 & 0.000 & n/a & 2.8 \\
\bottomrule
\end{tabular}
\end{table}

\begin{figure}[t]
  \centering
  \includegraphics[width=\linewidth]{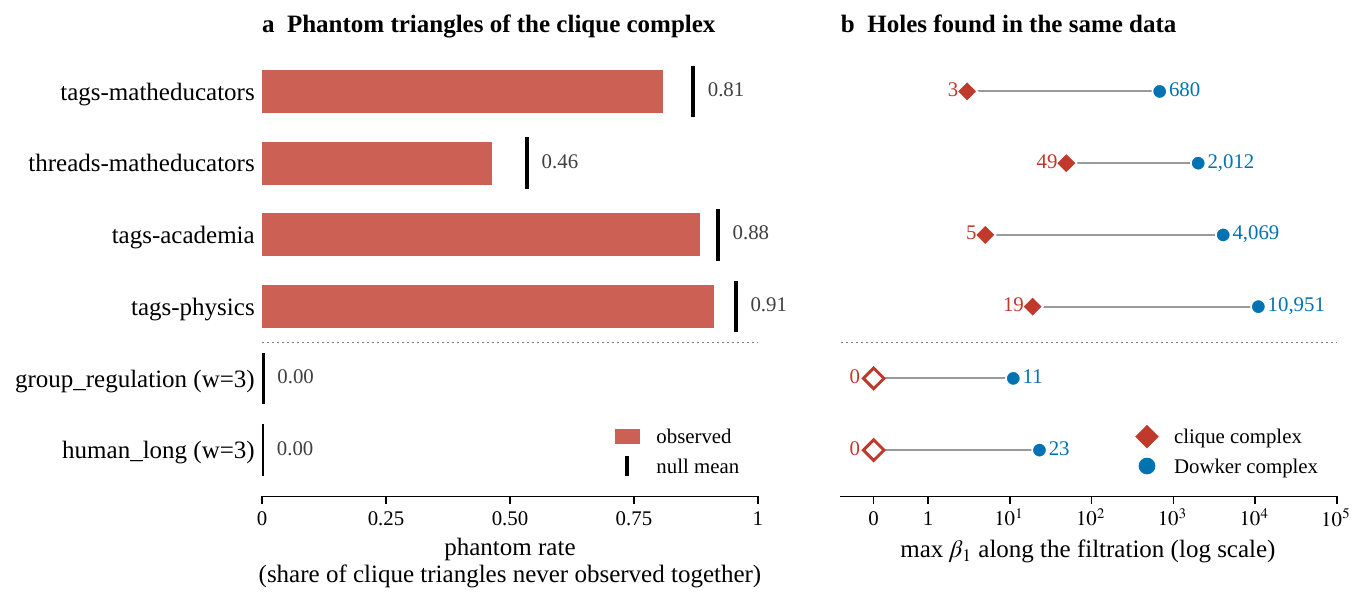}
  \caption{Phantom triangles and holes. (a) Phantom rate of the clique complex, i.e.\ the share of its triangles whose three vertices never co-occurred (bars), with the mean under the degree- and size-preserving null model (black ticks). (b) Maximum $\beta_1$ along the filtration for the clique (red diamonds) and Dowker (blue circles) complexes; open diamonds mark $\beta_1=0$. Above the dotted line: Stack Exchange; below: TNA example data ($w=3$).}
  \label{fig:phantom}
\end{figure}

\begin{figure}[t]
  \centering
  \includegraphics[width=\linewidth]{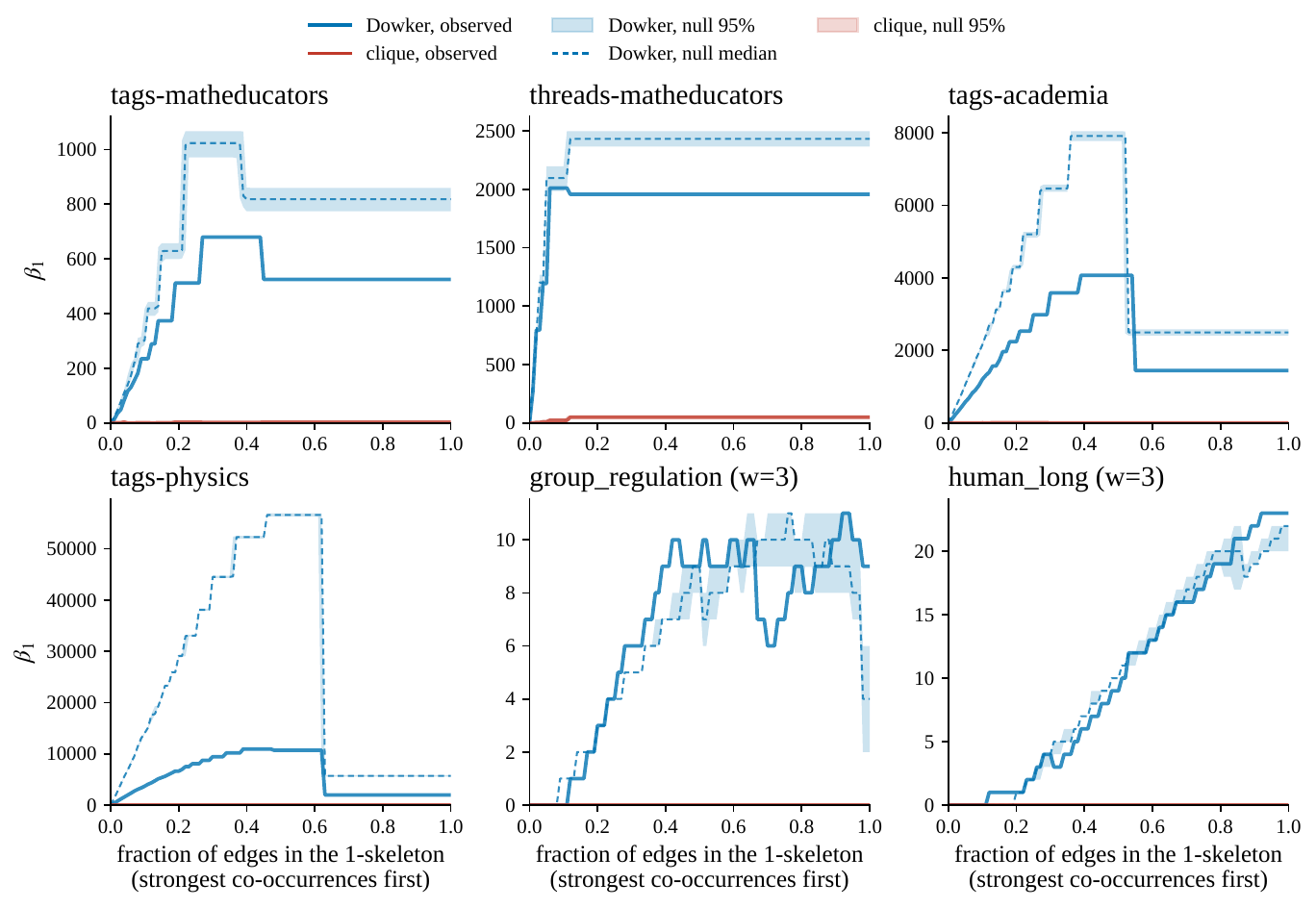}
  \caption{$\beta_1$ along the count filtration for the Dowker (blue) and clique (red) complexes, with the 95\% band and median of the null model. The horizontal axis is the fraction of 1-skeleton edges present, strongest co-occurrences first.}
  \label{fig:curves}
\end{figure}

\begin{figure}[t]
  \centering
  \includegraphics[width=\linewidth]{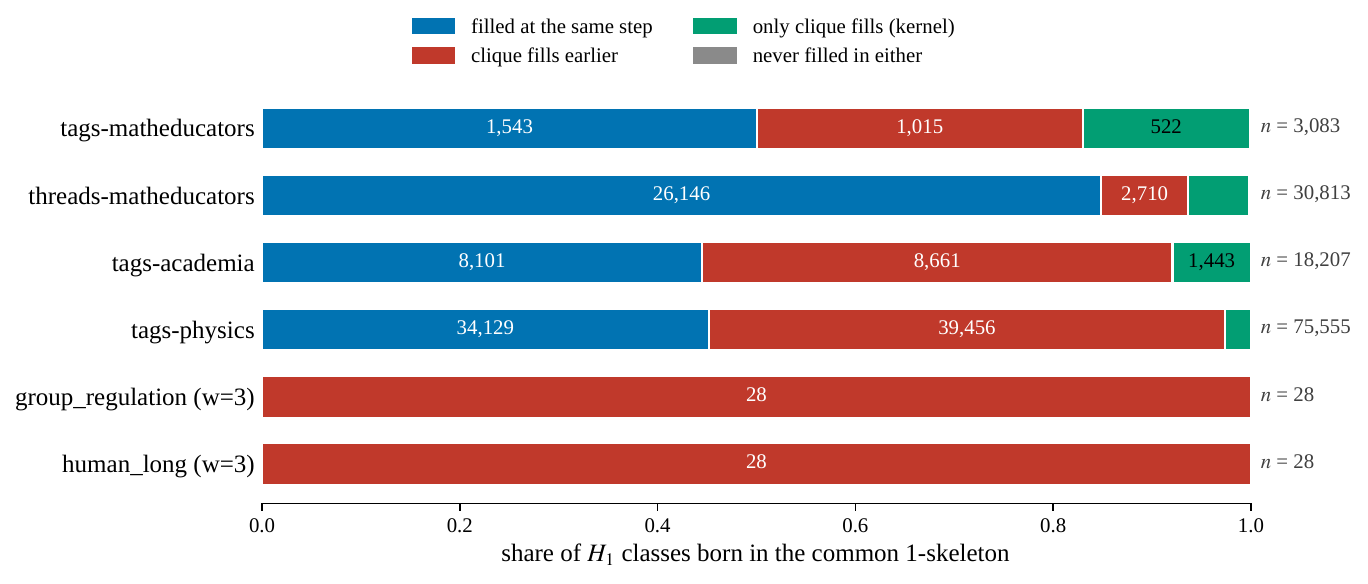}
  \caption{Fate of the $H_1$ classes born in the common 1-skeleton.}
  \label{fig:classes}
\end{figure}

\begin{figure}[t]
  \centering
  \includegraphics[width=0.7\linewidth]{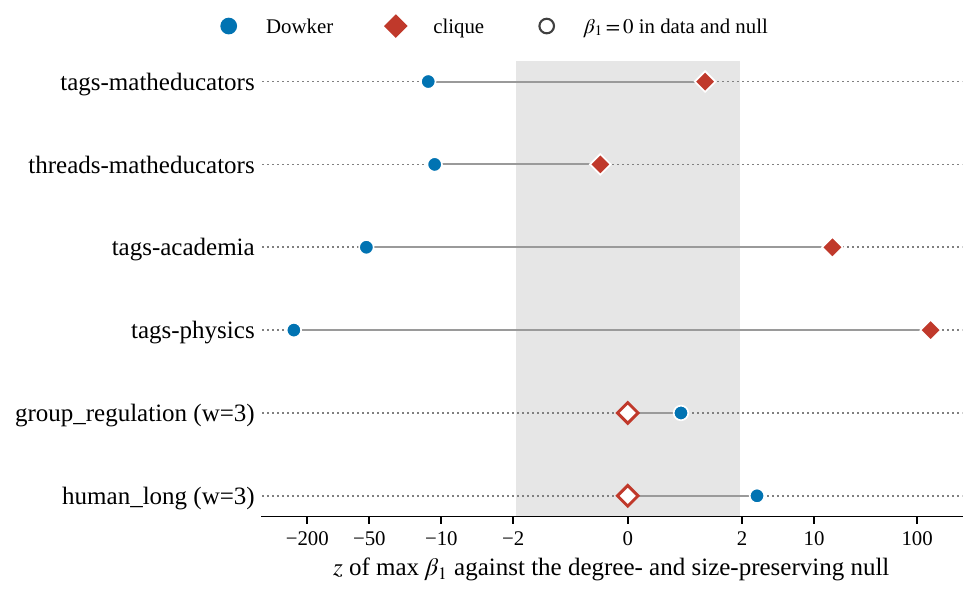}
  \caption{$z$-scores of the maximum $\beta_1$ against the null model. The grey band marks $|z|<1.96$. Open markers: $\beta_1=0$ in the data and in every null sample.}
  \label{fig:z}
\end{figure}

This section reports the measurements obtained with the methods of \Cref{sec:data}. We first show that in the Stack Exchange group data most clique triangles are phantom and many first-homology classes are erased, a permanent loss (\Cref{tab:main}). We then show that on the TNA example data no triangle is phantom, yet the clique complex shows no first homology at all, because it fills every hole that the Dowker complex opens partway through the filtration, a transient loss. Next we show that the two constructions lead to different conclusions against the null model, and that neither form of loss is visible at fixed thresholds. Finally, we report two further behaviours of the tooling found in the course of the analysis.

\paragraph{Phantom triangles are the majority in group data.}
In the four Stack Exchange data sets the phantom rate---the share of triangles of the full clique complex whose three vertices never co-occurred---is 0.46--0.91 (\Cref{fig:phantom}a and \Cref{tab:main}).
In \texttt{tags-physics}, where it is highest (0.91), roughly nine in ten triangles filled by the clique complex correspond to no observed group.
The Dowker complex, by contrast, fills only triples that were actually observed, so its phantom rate is zero by definition: it never adds a 2-simplex for a triple that does not occur in the data.
Under a degree- and size-preserving null model the phantom rate of the clique complex is even higher (0.534--0.955; \Cref{fig:phantom}a and \Cref{tab:null}), so the abundance of phantom triangles is not an accident of these data sets but a consequence of applying the clique construction to group data at all.
They are not harmless: the clique complex fills 522--1{,}970 first-homology classes that the Dowker complex keeps open to the end of the filtration (the kernel of \Cref{prop:erase}; the permanent loss), and fills a further 1{,}015--39{,}456 classes earlier than the data warrant (the transient loss; \Cref{fig:classes}).
Along the filtration the Dowker complex reaches $\beta_1$ values of 680--10{,}951, whereas the clique complex never exceeds 49 (\Cref{fig:phantom}b and \Cref{fig:curves}).

\paragraph{On the TNA example data the loss is transient but complete.}
With nine codes, every triple eventually co-occurs, so the phantom rate of the full complex is zero in both data sets.
The difference lies in \emph{when} triangles are filled.
The clique complex has $\beta_1=0$ at every threshold: a triangle is filled the moment its last edge appears.
The Dowker complex opens up to 11 (\texttt{group\_regulation}) and 23 (\texttt{human\_long}) cycles along the filtration---triples whose pairs co-occur strongly but which have not yet co-occurred as triples---and all 28 first-homology classes are filled earlier by the clique complex than by the Dowker complex (\Cref{tab:main,fig:classes}).
The loss is thus entirely transient: no class is erased for good, yet none of the holes is visible in the clique complex.

\paragraph{Against the null model, the two constructions tell different stories.}
That the clique complex has fewer cycles is guaranteed by \Cref{prop:erase} ($\beta_1^{\mathrm{C}}\le\beta_1^{\mathrm{D}}$); the question here is whether the two constructions lead to the same conclusion about the data.
The Dowker complex departs strongly from the null in all four Stack Exchange data sets (\Cref{tab:null,fig:z}): observed groups produce far fewer cycles than randomized groups with the same sizes and degrees ($z$ from $-11.6$ to $-270.1$), i.e.\ the observed groups are more closed than random.
On the TNA data the departures are small and point the other way: $z=0.9$ for \texttt{group\_regulation} and $2.8$ for \texttt{human\_long} with $w=3$ ($3.5$ and $0.5$ with $w=4$; \Cref{app:w4}).
The clique complex cannot show any of this on the TNA data, where $\beta_1=0$ in the data and in every null sample, and it shows no significant departure on the two Math Educators data sets.
On \texttt{tags-academia} and \texttt{tags-physics} the clique complex does depart from the null, but the departure concerns 5 and 19 cycles against thousands for the Dowker complex, and it points the other way (more cycles than random).
This null model also destroys pairwise structure, so ``more closed than random'' is not by itself evidence of irreducible many-body integration (\Cref{sec:limits}); what the comparison shows is that the two constructions lead to different conclusions about the same data.

\paragraph{At fixed thresholds the discrepancy is invisible.}
In practice \texttt{build\_simplicial} is applied at a single threshold to the symmetrized transition matrix.
At the thresholds 0.02--0.2 the two complexes have the same triangles and the same $\beta_1$ in both TNA data sets (\Cref{tab:fixed}).
Along the count filtration, however, the Dowker complex opens up to 11 and 23 cycles that the clique complex fills earlier (\Cref{tab:main}).
These cycles exist only in the middle of the filtration, where clique triangles enter before their triples have co-occurred, so a fixed-threshold analysis cannot reveal the transient loss.

\begin{table}[t]
\centering
\caption{Fixed-threshold analysis on the symmetrized transition matrices (as with \texttt{Nestimate::build\_simplicial}), windows $w=3$. Clique triangle counts match the package output exactly.}
\label{tab:fixed}
\small
\begin{tabular}{@{}lrrrrrr@{}}
\toprule
Data & threshold & edges & clique $\triangle$ & Dowker $\triangle$ & phantom & $\beta_1$ (C / D) \\
\midrule
group\_regulation & 0.02 & 34 & 70 & 70 & 0 & 0 / 0 \\
                  & 0.05 & 30 & 50 & 50 & 0 & 0 / 0 \\
                  & 0.10 & 22 & 20 & 20 & 0 & 0 / 0 \\
                  & 0.20 & 13 & 3  & 3  & 0 & 2 / 2 \\
human\_long       & 0.02 & 36 & 84 & 84 & 0  & 0 / 0 \\
                  & 0.05 & 36 & 84 & 84 & 0  & 0 / 0 \\
                  & 0.10 & 23 & 22 & 22 & 0  & 0 / 0 \\
                  & 0.20 & 7  & 0  & 0  & 0  & 0 / 0 \\
\bottomrule
\end{tabular}
\end{table}

\paragraph{Two further behaviours of the tooling.}
We note two behaviours of \texttt{Nestimate} 0.8.5 that matter for simplicial analyses.
First, \texttt{build\_network(method = "co\_occurrence", window\_size = k)} returned identical matrices for $k=2,3,4$: pairs are counted over whole sequences with multiplicity, so the window size does not enter.
Second, \texttt{persistent\_homology} reports the Betti numbers of the complex truncated at \texttt{max\_dim}; in its top dimension these are artefacts of the truncation (for example $\beta_2=56=\binom{8}{3}$ for the complete graph on nine codes, whose clique complex is contractible).
Both are easy to address; we mention them because they affect the interpretation of published analyses.

\section{Discussion}\label{sec:discussion}

This section discusses what the results mean for learning-analytics practice (\Cref{sec:implications}), sets out the limitations of the study (\Cref{sec:limits}) and outlines future work (\Cref{sec:outlook}).

\subsection{Implications for learning-analytics practice}\label{sec:implications}
Our results suggest three practical points for simplicial analyses of learning data.

\paragraph{Choose the complex from the data type.}
The choice of complex should be made from the data type, before any threshold is chosen (\Cref{tab:rules}).
Much of the data in learning analytics comes as observed groups: windows of coded discourse, forum threads, records of group activity.
Once such data are reduced to a pairwise co-occurrence matrix, the information about which triples actually occurred together is lost before any complex is built, and no later choice of threshold or pipeline can recover it (\Cref{prop:erase}).
When the input consists of observed groups, the groups themselves should therefore be kept as the input and a Dowker complex built from them.

\paragraph{Agreement at a single threshold does not validate the construction.}
On the TNA example data the two constructions have the same $\beta_1$ at every threshold used in practice (\Cref{tab:fixed}).
Along the whole filtration, however, the clique complex closes every cycle that the Dowker complex opens, so only the Dowker complex can be compared with the null model ($\beta_1^{\mathrm{C}}=0$ in the data and in every null sample; \Cref{sec:results}).
Because these cycles exist only partway through the filtration, before the corresponding triples co-occur, a comparison at particular thresholds is unlikely to reveal the difference.
Checking the construction requires comparing whole filtrations and comparing both against a null model.

\paragraph{Report both forms of loss as diagnostics.}
The phantom rate needs only a count of the clique triangles and of the triples actually observed, and the kernel dimension follows from computing the persistent homology of both complexes from the same input.
Neither requires additional data, and both are cheap.
Reported alongside clique-based results, they let readers judge how much of the reported topology depends on the choice of construction.
These two, however, measure only the permanent loss, and a phantom rate of zero does not mean that the two constructions agree: both TNA data sets have a phantom rate of 0.00, yet in each of them all 28 first-homology classes are filled earlier by the clique complex, which shows no hole at all (\Cref{tab:main}).
The number of classes filled earlier by the clique complex, which measures the transient loss, should therefore be reported together with the phantom rate and the kernel dimension.
Which form dominates depends on the data: in our data, groups over many entities (Stack Exchange) showed a large permanent loss, whereas short windows over nine codes (TNA), in which every triple eventually co-occurs, showed only the transient loss.
Before any homology is computed, \Cref{prop:conformal} gives a direct check: if the facets at a threshold form a conformal hypergraph, the two constructions coincide there, and each failure of Gilmore's criterion points to the groups responsible.

\subsection{Limitations}\label{sec:limits}
This study has five limitations.
\begin{enumerate}
  \item[(i)] \textbf{Design of the null model.} Our null model preserves group sizes and the degree of each entity but also destroys the pairwise co-occurrence structure. That the observed data have fewer cycles than random, i.e.\ are ``more closed than random'', is therefore not by itself evidence of irreducible many-body integration. What the comparison shows is that, on the same data and against the same null model, the two constructions lead to different conclusions. Testing irreducibility requires a test that conditions on all pairwise co-occurrences (\Cref{sec:outlook}).
  \item[(ii)] \textbf{Independence of observations.} In the TNA data, consecutive windows are cut from the same sequence, so neighbouring windows share events and the observations are not independent. The $z$-scores and $p$-values against the null model should therefore be read as descriptive measures of the size of the departure, not as exact tests that assume independent observations. Tests that account for this dependence, for example by resampling whole sequences, are left to future work.
  \item[(iii)] \textbf{Dependence on the window and synthetic data.} On the TNA data the departure of the Dowker complex from the null is small and depends on the window: $z=0.9$ with $w=3$ and $3.5$ with $w=4$ for the synthetic \texttt{group\_regulation}, and $2.8$ and $0.5$ for \texttt{human\_long} (\Cref{app:w4}). Being synthetic, \texttt{group\_regulation} also need not reflect the group structure of real collaborative learning. Guidance on choosing the window and validation on real collaborative-learning data are left to future work.
  \item[(iv)] \textbf{Domain of the data.} The Stack Exchange data come from online question-and-answer communities, not from classrooms. Math Educators, a community of mathematics teachers, is closest to an educational setting, but it is still not a record of collaboration among learners. In particular, the permanent loss appeared only in the Stack Exchange data; whether it occurs in classroom data remains to be tested. We also do not relate topology to learning outcomes. Applying the approach to classroom discourse and records of collaborative learning, and examining how the resulting topology relates to learning processes and outcomes, are left to future work.
  \item[(v)] \textbf{No higher-order dependencies.} We study many-body co-occurrence (HOI) only, not memory in sequences (HOD), which \texttt{Nestimate} already covers extensively. \Cref{app:hodge} relates transition flows to the co-occurrence complex, but only through the net flow, not memory, and without data. Treating both kinds of higher order in a single framework is beyond the scope of this paper.
\end{enumerate}

\subsection{Outlook}\label{sec:outlook}
This study has shown that observed groups call for the Dowker complex rather than the clique complex. Taking this as a starting point, the work can be extended in seven directions.
Direction (1) carries the results over to existing tooling; (2)--(5) extend the analysis to irreducibility, time, heterogeneous interaction data and direction; (6) examines systematically when the two constructions diverge most; and (7) interprets the holes in terms of learning.
Of these, (2) addresses limitation (i), (3) corresponds to the ``Time-stamped groups'' row of \Cref{tab:rules}, and (5) to its ``Transitions (directed)'' row.
\begin{enumerate}
  \item[(1)] \textbf{A bridge in existing tooling.} Observed groups are already representable in \texttt{Nestimate} (\texttt{bipartite\_groups}), and \texttt{hypernets} now builds a window co-occurrence complex. Building Dowker complexes by containment counts from any observed groups, and reporting the diagnostics of this paper next to clique-based results, would let users choose the construction according to the data type within the same pipeline; concrete suggestions for the current packages are given in \Cref{app:suggest}.
  \item[(2)] \textbf{Existence and irreducibility.} The Dowker complex records whether a set of entities occurred together; whether it occurred more often than its pairs imply is a separate, statistical question. We propose a two-stage analysis: Dowker complexes for existence, followed by a test of three-way interaction that conditions on all pairwise co-occurrences, as in log-linear models without a three-way term~\citep{bishop1975discrete}; see \citet{musciotto2021detecting} for a related validation of hyperedges against a degree-preserving null. Results will be reported separately.
  \item[(3)] \textbf{Time.} We have previously segmented learning into periods from changes in the connected components ($\beta_0$) of zigzag persistent homology~\citep{inoue2026}. Cumulative and zigzag~\citep{carlsson2010zigzag} Dowker filtrations would follow not only connected components but also when groups form and when cycles close. This corresponds to the ``Time-stamped groups'' row of \Cref{tab:rules}.
  \item[(4)] \textbf{Heterogeneous interaction data.} Learner-by-behaviour data as in HINA~\citep{feng2026hina} are bipartite relations; by \Cref{prop:duality}, a Dowker analysis gives the same persistent homology from the learner side and from the behaviour side in every dimension, whereas the two clique projections agree only up to $H_1$ and both erase the kernel.
  \item[(5)] \textbf{Direction.} As shown in \Cref{sec:rules}, direction accounts for a large part of the TNA transition-probability matrices, but current pipelines discard it by symmetrization. Directed flag complexes~\citep{lutgehetmann2020computing} and path homology~\citep{grigoryan2012homologies} compute topology while keeping direction, and so offer a way to study the higher-order structure of transition data themselves. A complementary route keeps the co-occurrence complex and treats the net flow of transitions as an edge flow on it; \Cref{app:hodge} shows that phantom triangles then turn harmonic circulation into local curl, and comparing the harmonic share of the flow under the two constructions is left to future work.
  \item[(6)] \textbf{Decision rules.} \Cref{prop:conformal} says exactly when the two constructions agree; what remains is to say by how much they disagree. Which data properties (density, group-size distribution, overlap between groups, the number of entities relative to group size) control the permanent and the transient loss will be mapped in a separate study, with bounds where possible. This would let analysts estimate in advance, from the properties of their data, how much the choice of construction matters.
  \item[(7)] \textbf{Interpreting holes in learning terms.} A first-homology class of the Dowker complex is a ring of ideas, codes or participants that are connected pair by pair but were never brought together in one group. In Knowledge Building~\citep{scardamalia2006knowledge}, such a ring may mark ideas that a community has related only piecemeal and has yet to rise above; in terms of group cognition~\citep{stahl2006group}, it marks connections that exist only between pairs and not at the level of the group. Representative cycles from classroom discourse, for example data analysed with KBDeX~\citep{oshima2012kbdex}, read together with learning scientists, would test these interpretations and connect the topology to learning processes and outcomes (limitation (iv)).
\end{enumerate}

\section{Conclusion}\label{sec:conclusion}
A clique complex built from pairwise co-occurrence is a function of the pairwise network; on data that consist of observed groups it can only erase their topology, never add to it.
On public co-occurrence data the clique construction erases many holes for good and changes the conclusions drawn against a null model. On the example data of the TNA toolchain, with the tools' own output reproduced exactly, it erases no hole for good but fills every hole before the data do, so that it shows none. Neither loss is visible at the fixed thresholds used in practice, and both can be reported as diagnostics.
Whether this happens can be checked on triples of groups.
The construction should follow the data type.
For observed groups that construction is the Dowker complex, which fits into existing pipelines and fills the one piece they lack: the topology of the groups themselves.

\section*{Code and data availability}\label{sec:availability}
Code and scripts that regenerate every figure and table are available at \url{https://github.com/igu-lab/beyond-the-clique}.
The code and results for this version of the paper are archived at Zenodo: \url{https://doi.org/10.5281/zenodo.23147037} (all versions: \url{https://doi.org/10.5281/zenodo.23096088}).
The data themselves are not redistributed; they can be obtained from the sources below and converted with the scripts in the repository, and the sizes in \Cref{tab:main} can be used to check the converted data.
The Stack Exchange data are derived from the \texttt{Posts.xml} files of the Mathematics Educators, Academia and Physics sites in the public data dump of 31 December 2025~\citep{stackexchange2025dump}, licensed under CC BY-SA.
The TNA example data are distributed with the R packages \texttt{tna} 1.2.3 (\texttt{group\_regulation}) and \texttt{Nestimate} 0.8.5 (\texttt{human\_long}~\citep{saqr2026vibe}).
Preprocessing and reproduction steps are given in \Cref{app:repro}.

\section*{Acknowledgments}\label{sec:ack}
This work was supported by JSPS KAKENHI JP26K00528, JP26K22187, JP25K00845, JP25K21951.

\bibliographystyle{unsrtnat}
\bibliography{refs}

\appendix
\section{Proofs and additional results}\label{app}

\subsection{Persistence statement of \Cref{prop:erase}}\label{app:persistence}
Index the count filtration by decreasing $t$. Every edge enters both complexes at the same value $w(e)$, so a first-homology class, which is born when an edge closes a cycle in the common 1-skeleton, is born at the same value in both. A class dies when the cycle becomes a boundary. Since $B_1(\KD(t))\subseteq B_1(\KC(t))$ for every $t$, a cycle that bounds in $\KD(t)$ bounds in $\KC(t)$, so its death in $\KC$ is no later than in $\KD$. Classes that never die in $\KD$ but die in $\KC$ are exactly the elements of $\ker\iota_*$ at the end of the filtration. We pair classes by the positive edge that creates them in the common 1-skeleton, using the same column reduction for both complexes.

\subsection{\Cref{prop:duality}: sketch and example}\label{app:duality}
Let $R\subseteq X\times Y$. The Dowker complex $D_X$ has as simplices the finite $\sigma\subseteq X$ with a common $y$ such that $(x,y)\in R$ for all $x\in\sigma$; $D_Y$ is defined symmetrically. Dowker's theorem~\citep{dowker1952} states $D_X\simeq D_Y$; \citet{chowdhury2018dowker} prove a functorial version that yields equal persistence diagrams for filtered relations. The clique complexes $C_X$ and $C_Y$ of the one-mode projections are the flagifications of $D_X$ and $D_Y$. \citet{huber2025flagifying} prove that their persistent homology agrees in dimensions 0 and 1 and show by the example below that it can differ in dimension 2.

\emph{Example} (the counterexample of \citet[Prop.~3.4]{huber2025flagifying}, stated as a relation). Let $Y=\{1,2,3,4\}$ be four learners and $X$ the six two-person conversations $\{1,2\},\dots,\{3,4\}$, with $(x,y)\in R$ when learner $y$ took part in conversation $x$. $D_Y$ is the complete graph $K_4$ without triangles, so $\beta_1(D_Y)=3$, while $C_Y$ is the full tetrahedron and is contractible. $D_X$ has four triangles (the three conversations of each learner) on the octahedral 1-skeleton (two conversations are adjacent when they share a learner), so $\beta_1(D_X)=12-6+1-4=3$, in agreement with $D_Y$. Its clique complex $C_X$ fills all eight triangles of the octahedron and is a 2-sphere, $\beta_2(C_X)=1$. Hence $C_X\not\simeq C_Y$, while $D_X\simeq D_Y$. A random search over relations with up to six elements on each side found no smaller example.

\subsection{Windows of length four}\label{app:w4}
\begin{table}[h]
\centering
\caption{TNA example data with windows $w=4$ (21{,}867 and 9{,}530 windows). Columns as in \Cref{tab:main,tab:null}.}
\label{tab:w4}
\small
\begin{tabular}{@{}lrrrrrr@{}}
\toprule
Data & phantom rate (null) & max $\beta_1$ C / D & null D & $z$ (D) & kernel & earlier \\
\midrule
group\_regulation ($w{=}4$) & 0.00 (0.00) & 0 / 8 & 6.3 & 3.5 & 0 & 28 \\
human\_long ($w{=}4$) & 0.00 (0.00) & 0 / 17 & 16.8 & 0.5 & 0 & 28 \\
\bottomrule
\end{tabular}
\end{table}
With $w=4$ the qualitative picture is unchanged: no triangle is phantom, the clique complex has $\beta_1=0$ throughout and fills all 28 classes earlier, and neither data set departs strongly from the null (\texttt{group\_regulation} $z=3.5$, \texttt{human\_long} $z=0.5$, both towards more cycles than random).

\subsection{Preprocessing and reproduction}\label{app:repro}
Stack Exchange data sets were converted from the \texttt{Posts.xml} files of the dump into the nverts/simplices/times format of \citet{benson2018simplicial}, with no cap on group size. TNA sequences were exported from R in long format (actor, time, action) and converted to windows within each sequence; a sequence shorter than $w$ forms a single group. The clique construction of \texttt{Nestimate::build\_simplicial} was re-implemented and checked against the package output for both data sets at thresholds 0.02, 0.05, 0.10 and 0.20: the simplex sets are identical. The repository layout and the commands that regenerate every figure and table are documented in the README of the repository (see Code and data availability).

\section{Relation to the TNA toolchain (as of 2 October 2026)}\label{app:tools}
This appendix relates our results to the R toolchain around TNA (\texttt{tna}, \texttt{Nestimate}, \texttt{hypernets}), whose higher-order functions are developing rapidly. It records the state of the packages on 2 October 2026; the main text does not depend on it.

\subsection{The window complex of \texttt{hypernets}}\label{app:hypernets}
Since 1 October 2026 the development version of \texttt{hypernets}~\citep{hypernets060}, which builds hypergraph methods on top of \texttt{Nestimate}, has offered \texttt{simplicial(type = "window")}, a co-occurrence complex in the sense of \citet{salnikov2018cooccurrence} whose groups are the sets of distinct actions in windows of consecutive actions---the same kind of groups as in our analysis of the TNA data (\Cref{sec:data}).
Sets can also be validated against a swap null model that, like ours, preserves group sizes and entity degrees, or against the hypergeometric null of \citet{musciotto2021detecting}.
With \texttt{min\_count = 1} this complex is the Dowker complex $\KD(1)$ of the windows.
Its count filtration, however, counts how often each \emph{identical} window set recurs, and a simplex enters at the largest such count among the window sets that contain it, whereas the Dowker count $w(\sigma)$ counts every window that contains $\sigma$.
The two filtrations therefore differ from $t=2$ on: for the windows $\{a,b,c\}$ and $\{a,b,d\}$, the edge $\{a,b\}$ has $w=2$ and lies in $\KD(2)$, but neither window set recurs, so in the \texttt{hypernets} filtration the edge enters only at count~1. We confirmed this with version 0.6.0 (commit cd4ddcf) on these two windows: with \texttt{min\_count = 1} the edge $\{a,b\}$ has count~1, and with \texttt{min\_count = 2} only the four vertices remain, since vertices are kept at every count.
In general the \texttt{hypernets} complex at count $t$ is a subcomplex of $\KD(t)$ whose 1-skeleton can be smaller than the thresholded co-occurrence network $G(t)$, so it cannot be compared with the clique complex of $G(t)$ on a common 1-skeleton.
The package does not compare its window complex with the clique construction, and it is built from the action side only, whereas the dual window side is equally available (\Cref{prop:duality}).

\subsection{Transitions as flows on the co-occurrence complex}\label{app:hodge}
TNA models transitions rather than co-occurrences, and the main text treats co-occurrence only (\Cref{sec:rules}). The loss of \Cref{prop:erase} nevertheless bears on transition data, through the Hodge decomposition of edge flows~\citep{jiang2011statistical,lim2020hodge,schaub2020random}.

Orient the edges of $G(t)$ and give the space $C^1$ of real edge functions the standard inner product. A transition matrix $W$ on the same entities defines the net flow $f(u,v)=W_{uv}-W_{vu}$, the antisymmetric part measured in \Cref{sec:rules}. For a complex $K$ with 1-skeleton $G(t)$, $C^1=\mathrm{grad}\oplus\mathrm{curl}^K\oplus\mathcal{H}^K$ orthogonally, where $\mathrm{grad}$ depends only on $G(t)$, $\mathrm{curl}^K$ is spanned by the boundaries of the triangles of $K$, and $\dim\mathcal{H}^K=\beta_1(K)$~\citep{lim2020hodge}; in HodgeRank terms~\citep{jiang2011statistical}, these are a global ordering of the entities, circulation around triangles, and circulation around holes.

\begin{corollary}[of \Cref{prop:erase}: phantom triangles turn harmonic flow into curl]\label{cor:hodge}
Let $\mathcal{L}=\mathcal{H}^{\mathrm{D}}\cap\mathrm{curl}^{\mathrm{C}}$. Then $\mathcal{H}^{\mathrm{D}}=\mathcal{H}^{\mathrm{C}}\oplus\mathcal{L}$, $\dim\mathcal{L}=\beta_1^{\mathrm{D}}(t)-\beta_1^{\mathrm{C}}(t)=\dim\ker\iota_*$, and $\mathcal{L}$ is spanned by the harmonic parts, in $\KD(t)$, of the boundaries of the phantom triangles. Hence, for every flow $f$, the clique complex moves exactly the component $P_{\mathcal{L}}f$ from the harmonic to the curl part, and the harmonic parts satisfy $\|h^{\mathrm{C}}(f)\|\le\|h^{\mathrm{D}}(f)\|$.
\end{corollary}

\begin{proof}
Since $\KD(t)\subseteq\KC(t)$ with a common 1-skeleton, $\mathrm{grad}$ is common and $\mathrm{curl}^{\mathrm{D}}\subseteq\mathrm{curl}^{\mathrm{C}}$. The orthogonal complement of $\mathrm{curl}^{\mathrm{D}}$ in $\mathrm{curl}^{\mathrm{C}}$ is orthogonal to $\mathrm{grad}\oplus\mathrm{curl}^{\mathrm{D}}$, so it lies in $\mathcal{H}^{\mathrm{D}}$ and equals $\mathcal{L}$; thus $C^1=\mathrm{grad}\oplus\mathrm{curl}^{\mathrm{D}}\oplus\mathcal{L}\oplus\mathcal{H}^{\mathrm{C}}$, with $\mathcal{L}$ counted in $\mathcal{H}^{\mathrm{D}}$ for $\KD(t)$ and in $\mathrm{curl}^{\mathrm{C}}$ for $\KC(t)$. This gives the splitting and the last statement, and $\dim\mathcal{L}=\beta_1^{\mathrm{D}}-\beta_1^{\mathrm{C}}=\dim\ker\iota_*$ by \Cref{prop:erase}. Finally $\mathrm{curl}^{\mathrm{C}}=\mathrm{curl}^{\mathrm{D}}+\operatorname{span}\{\partial\tau:\tau\text{ phantom}\}$, and projecting onto $\mathcal{L}$ gives the spanning set.
\end{proof}

In words: circulation of actions around a cycle that the learners never closed as a group---for example $a\to b\to c\to a$ when $\{a,b,c\}$ is a phantom triangle whose boundary no observed groups fill---has a harmonic component in the Dowker complex, reflecting a hole in the observed co-occurrence; the clique complex reports this component as local curl, as if the circulation took place inside a group.
The corollary uses only the net flow, not the Markov structure of TNA; it is a conceptual link between the two meanings of ``higher order'' (\Cref{sec:intro}), and its empirical use is left to future work (\Cref{sec:outlook}).

\subsection{Suggestions for existing tooling}\label{app:suggest}
The development version 0.9.12 of \texttt{Nestimate}~\citep{nestimate0912} offers the same three simplicial types as version 0.8.5 (\Cref{sec:rules}).
The window complex of \texttt{hypernets}~\citep{hypernets060} is a first step: at \texttt{min\_count = 1} it is the Dowker complex of the windows. Three additions would make it a full Dowker construction: a count filtration by containment, $w(\sigma)$, so that it shares its 1-skeleton with the clique complex; groups other than windows, such as threads, sessions and the \texttt{bipartite\_groups} hypergraphs of \texttt{Nestimate}; and the diagnostics of this paper (phantom rate, classes filled earlier, kernel dimension and the check of \Cref{prop:conformal}) reported next to clique-based results. Users could then choose the construction according to the data type and see its consequences within the same pipeline.

\end{document}